\documentclass[conference]{IEEEtran}
\usepackage{bm}
\usepackage{amsmath}
\usepackage{amsthm}
\usepackage{amssymb}
\usepackage{graphicx}
\usepackage{url}
\usepackage[usenames,dvipsnames]{xcolor}
\usepackage{color}
\usepackage{cite}
\usepackage{algorithm}
\usepackage{algpseudocode}

\makeatletter
\theoremstyle{plain}
\newtheorem{thm}{\protect\theoremname}
\newtheorem{example}{Example}

\newcommand{\EC}[1]{#1}
\newcommand{\Ra}[1]{#1}
\newcommand{\Rb}[1]{#1}
\newcommand{\Rc}[1]{#1}
\newcommand{\Rd}[1]{#1}
\makeatother

\providecommand{\theoremname}{Theorem}

\begin{document}

\title{A Fisher Information Analysis of Joint Localization and Synchronization in Near Field}

\author{Henk Wymeersch\\
Department of Electrical Engineering \\Chalmers University of Technology \\e-mail: henkw@chalmers.se }
\maketitle
\begin{abstract}
In 5G communication, arrays are used for both positioning
and communication. As the arrays become larger, the far-field
assumption is increasingly being violated and curvature of the wavefront
should be taken into account. We explicitly contrast near-field and
far-field uplink localization performance in the presence of a clock bias from a Fisher information
perspective and show how a simple algorithm can provide a coarse estimate of 
a user's location and clock bias.
\end{abstract}

\section{Introduction}

Cellular localization has largely relied on measurements
of time-difference-of-arrival in either uplink or downlink.
Such measurements can cope with  the clock bias of the user, but
require multiple base stations (BSs) \cite{delPeralCOMST17}. With  5G, where large arrays are used to provide improved spectral efficiency,
angle measurements have become possible \cite{WymSecDesDarTuf:J18,buehrer2018collaborative}.
Estimating the angle-of-arrival (AOA) at different BSs,
the user's location is determined by a set of bearing lines, so that
localization can be performed without any stringent synchronization
requirements \cite{garcia2017direct}. When both user and base station
are equipped with large arrays, the user's position and orientation
can be inferred \cite{Shahmansoori2018}. 
Extensive studies have been performed to assess the \Rd{fundamental performance
of array-based positioning \cite{han2016performance,shen2010accuracy,Shen2010}
(through the Cram\'{e}r-Rao bound (CRB))}, as well as to develop practical
algorithms \cite{Shahmansoori2018,witrisal2016high}. All these works assume far-field propagation (plane wave assumption), where the user is far away from the BS. 
%
\EC{With communication systems beyond 5G targeting novel 
 technological enablers, such
as large intelligent surfaces and extreme aperture arrays \cite{bjornson2019massive}, 
the far-field propagation condition
may be violated, requiring us to revisit the models, performance characterization,
and algorithm design. Such activities have now started in communication
\cite{myers_message_2019,basar_wireless_2019} and radio localization
\cite{hu_beyond_2018} and are the main topic of this paper. } 

Near-field localization dates back  around 30
years in the context of source localization. The early work \cite{rockah1987array} studied
the impact of an imperfectly calibrated array on near-field source
localization and a calibration method was proposed, while
\cite{huang1991near} estimated the direction-of-arrival (DOA) of multiple
sources using the MUSIC algorithm and a maximum likelihood
(ML) approach. The latter was shown to be superior in low SNR conditions,
though comes at a significant complexity cost. In \cite{yuen_performance_1998},
an ESPRIT-based method was proposed and the performance was theoretically
determined. \EC{In \cite{gazzah2014crb}, a multi-source CRB was derived for stochastic sources, highlighting the benefits of centro-symmetric arrays. An overlapping sub-array approach was proposed in \cite{zhi2007near} for low complexity range and bearing estimation. }In \cite{el2010conditional}, time-varying
sources were studied in the narrowband regime in terms of the CRB. Range and bearing estimation were also treated in \cite{hu2014near}, based on sparse recovery techniques. 
\EC{The localization of near- and far-field sources was proposed in \cite{zuo2018subspace}.} A simplified CRB for near-field positioning was derived in \cite{Zhang19Spherical}, as well as an algorithm that  directly exploits the wavefront curvature for positioning. 
In contrast to the above works in the narrowband regime,
\cite{chen_maximum-likelihood_2002} considered spatial wideband signals \Rb{(where signals arrive at different antenna elements with different resolvable delays)} and derived the CRB and an ML estimator. The extension \cite{mada_efficient_2009}
relied on an expectation-maximization method, which is computationally
less demanding than ML.   Positioning using large intelligent
surface was considered in 
 \cite{hu_beyond_2018}, which showed that the CRB reduces quadratically in the size of the array. Most of these works rely on second-order statistics and are thus data-intensive. 

\EC{In this paper, in contrast to the above works, we consider not only the position but also the clock bias of the transmitter to be unknown. This leads us to investigate {joint localization and synchronization in near-field}. Moreover, we do not rely on second-order statistics and instead exploit the communication signal directly.  Our main contributions are: (i) a Fisher information analysis of uplink near-field joint localization
and synchronization with a linear array;  (ii) a simple joint localization and synchronization method using sub-array
processing.}

\section{System Model}
\begin{figure}
\begin{centering}
\includegraphics[width=0.8\columnwidth]{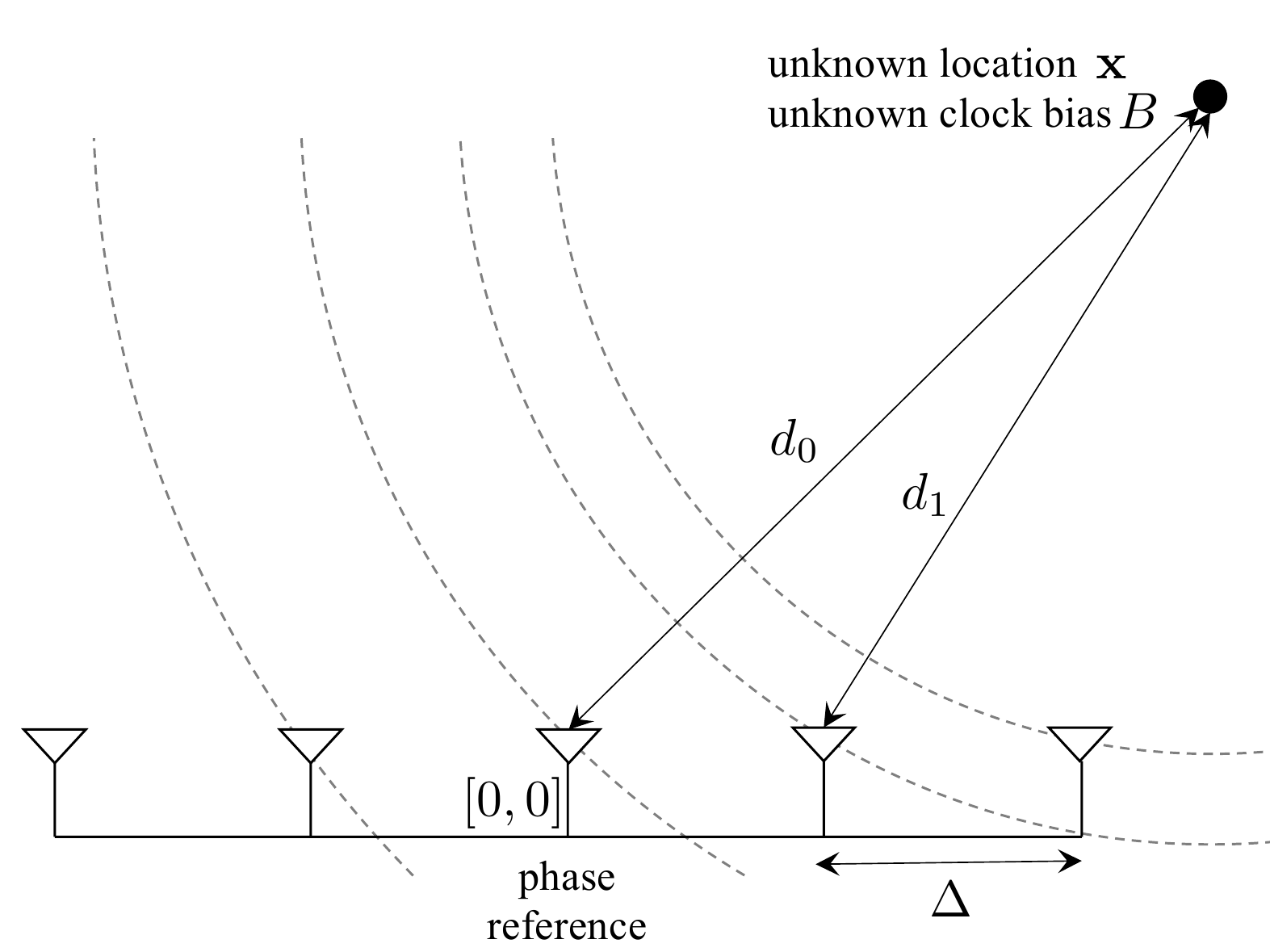}
\par\end{centering}
\caption{\label{fig:scenario}Scenario with a transmitting source and a receiver array. In near-field, the phase across the antenna elements changes nonlinearly. In spatial wide-band, the delay across the antenna elements changes noticeably. }
\end{figure}
We consider a 2D scenario with a single-antenna user equipment (UE) with unknown location $\mathbf{x}=[x,y]^{\text{T}}$
(or $[d,\theta]^{\text{T}}$ in polar coordinates, with $d=\Vert\mathbf{x}\Vert$
and $\theta=\arccos(x/\Vert\mathbf{x}\Vert)$) and a small BS (e.g., indoors in close proximity to the UE) with an $N+1$-element\footnote{The analysis can easily be modified for an array with $N$ elements,
with locations $\mathbf{x}_{n}=[\Delta/2+n\Delta,0]^{\text{T}}$,
for $n=-N/2,\ldots,N/2-1$. } linear array with element spacing $\Delta$, with locations $\mathbf{x}_{n}=[n\Delta,0]^{\text{T}},$
$n\in\{-N/2,\ldots,N/2\}$. The UE has an unknown clock bias $B$
(expressed in meters) and sends a known OFDM signal with transmit power
$P_{t}$ at a high carrier frequency $f_{c}$ (28 GHz or higher)
and a total bandwidth $W=(K+1)\Delta_{f}$, where $\Delta_{f}$ is
the subcarrier spacing and $K+1$ is the number of subcarriers. For
notational convenience, but without any loss of generality, we let
$k\in\{-K/2,\ldots,K/2\}$. We further introduce $d_{n}=\Vert\mathbf{x}-\mathbf{x}_{n}\Vert$ (so $d=d_0$)
and $\delta_{n}=\Vert\mathbf{x}-\mathbf{x}_{n}\Vert-B$, which allows
us to express the uplink signal observed on antenna $n$, subcarrier $k$:
\begin{align}
& y_{n}[k] =\label{eq:basicSignalModelOmega} \\ 
  & \alpha_{n}s[k]e^{-j\frac{2\pi}{\lambda}\xi_{n}[k]} +\sum_{l=1}^{L}\alpha_{n,l}s[k]e^{-j\frac{2\pi}{\lambda}\xi_{n,l}[k]}+w_{n}[k],\nonumber
\end{align}
where $w_{n}[k]$ complex
zero-mean Gaussian noise with variance $N_{0}/2$ per real dimension. Taking the phase at the center antenna ($N=0$) and center subcarrier $k=0$ as reference, i.e., $\xi_{0}[0]=0$, the phase at any antenna $n$ and any subcarrier $k$ is given by 
 \begin{align}
\xi_{n}[k]=(d_{n}-d_{0})+k\frac{\delta_{n}}{\left(K+1\right)T_{s}f_c}, \label{eq:signalPhaseGeneral}
\end{align}
in which
$T_{s}=1/W$. The first term is the $d_{n}-d_{0}$ is the difference in path length with respect to the center antenna, while the second term depends on the absolute delay $\delta_{n}$, and grows with the subcarrier index. 
\begin{example}
The phase $-2\pi \xi_{n}[k]/\lambda$ is shown in Fig.~\ref{fig:phasePlot}, for $\Delta=\lambda/2$, where the left figure shows the evolution of the phase with $n$ for various distances $d_0$ as a function of the antenna index $n$ and the right figure the evolution of the phase as a function of the subcarrier index $k$, for different antenna elements. We see from Fig.~\ref{fig:phasePlot} (left) that when $d_0$ is small, the phase exhibits a nonlinear behavior as a function of $n$, while Fig.~\ref{fig:phasePlot} (right) illustrates that for small $d_0$, different antenna elements see different absolute delays $\delta_{n}$ (note that this effect is only visible when we have a large bandwidth $W$). 
\end{example}

The complex channel gain at antenna $n$ is $\alpha_{n}=\rho_{n}e^{j\psi}$
with $\rho_{n}=\lambda/(2\pi d_{n})$ and $\psi=-2\pi d_{0}/\lambda$. Similarly, $\alpha_{n,l}$ is the complex gain of non-line-of-sight (NLOS) path $l$, $\xi_{n,l}[k]$ is a phase increasing with $k$ due to the delay of the NLOS path $l$.  
We will make several assumptions, in order to facilitate compact closed-form expressions: \EC{
we assume that $\alpha_{n}$ is not used directly for localization
(so it is treated as a separate unknown) \cite{Shahmansoori2018,win2018theoretical}; $|\alpha_{n}| \gg |\alpha_{n,l}|, \forall l$ so that  the line-of-sight (LOS) path is dominant  \cite{chen_maximum-likelihood_2002,el2010conditional,gazzah2014crb,myers_message_2019} (the robustness of the proposed method to multipath will be evaluated in Section \ref{sec:simsAlg})}; the transmitted signal spectrum is symmetric ($|s[k]|=|s[-k]|$).
\Rb{Our goal is to determine $\mathbf{x}$ and $B$ from the observation
$\mathbf{Y}\in\mathbb{C}^{(N+1)\times(K+1)}$, }
\Rd{though the proposed methods can be combined with tracking as in \cite{win2018theoretical} } to account for user mobility \Rd{and clock drift \cite{DarConFerGioWin:J09}. }

\subsubsection*{Terminology} 
\EC{
We distinguish  the following operating conditions. 
\begin{itemize}
\item \emph{Far-field vs near-field:} When $\Vert\mathbf{x}\Vert>2(N\Delta)^{2}/\lambda$,
the far-field regime with plane wave assumption holds. When $0.62 \sqrt{(N\Delta)^{3}/\lambda}< \Vert\mathbf{x}\Vert<2(N\Delta)^{2}/\lambda$, we operate in the radiative near-field zone, where wavefront curvature is non-negligible \cite{zuo2018subspace}. 
\item \emph{Narrowband vs wideband: }When $W  < c/(N\Delta)$, the signals
at the different antennas are not resolvable in the delay domain and communication is narrowband. When $W> c/(N\Delta)$, we consider the signals to be spatially wideband in the sense that they are  resolvable in the delay domain at different antennas \cite{chen_maximum-likelihood_2002}. 
\item \emph{Beam squint: } When $W> f_{c}/10$, the wavelength of the signal varies significantly over its bandwidth, leading to beam squint. Generally, beam squint  implies spatial wideband operation, but not vice versa.
\end{itemize}
Throughout this paper, we assume that the bandwidth is sufficiently small to ignore beam squint (i.e., $W\ll f_{c}$). }

\begin{figure}
\begin{centering}
\includegraphics[width=1\columnwidth]{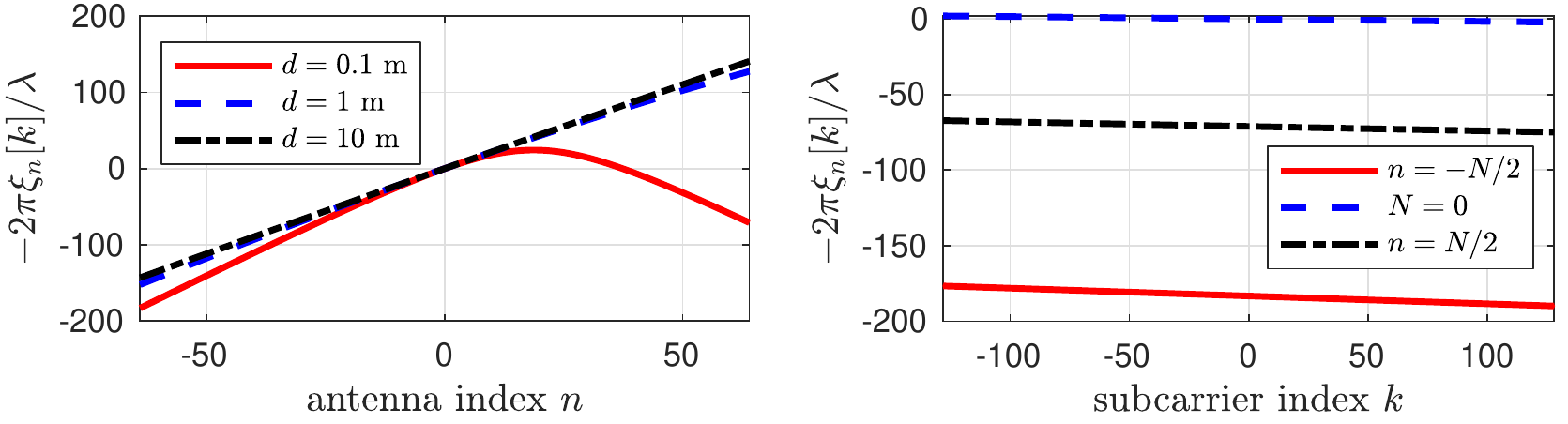}
\par\end{centering}
\caption{\label{fig:phasePlot}Left: Plot of the signal phase $-{2\pi}\xi_{n}[k]/{\lambda}$ as a function of the antenna index $n$ for different distances $d$ between transmitter and the phase reference of the receiver. Right: Plot of the signal phase $-{2\pi}\xi_{n}[k]/{\lambda}$ as a function of the subcarrier index $k$ for different antenna indices, for $d=0.1$ m. Parameters: $K=257$, $N=129$, $f_c=28~\text{GHz}$, $W=1.4~\text{GHz}$, $\Delta=\lambda/2$.}
\end{figure}

\section{Standard Far-Field and Narrowband Condition}

\subsection{Model and Fisher Information Matrix}

In far-field and narrowband condition, the 
standard model reverts to \cite{huang1991near,Shahmansoori2018}
\begin{align}
y_{n}[k] & =\alpha_0 s[k]e^{-j\frac{2\pi}{\lambda}\xi_{n}[k]}+w_{n}[k],\label{eq:SimplifiedModel}
\end{align}
where \Ra{
\begin{equation}
\xi_{n}[k]=-n\Delta\cos\theta-k(d-B)r_{f},\label{eq:rotationCase1}
\end{equation}}
with $r_{f}=\Delta_{f}/f_{c}$. This model is derived by taking a first-order Taylor series expansion  of $d_{n}-d_{0}$ around $n\Delta/d_{0}=0$ \cite{huang1991near} in combination with $\delta_{n}\approx\delta_{0}$
and  $\alpha_{n} \approx \alpha_{0}$.
%
The Fisher information matrix (FIM) of $\bm{\eta}=[\rho,\psi,d,\theta,B]^{\text{T}}$
is composed of the sum of the FIM for each subcarrier and each antenna
$\mathbf{J}^{\text{S}}(\bm{\eta})=\sum_{n=-N/2}^{N/2}\sum_{k=-K/2}^{K/2}\mathbf{J}_{n}[k]$
where \cite{Kay1993}
\begin{align}
 & \mathbf{J}_{n}[k]=\label{eq:FIMexpression}\\
 & \frac{1}{N_{0}}\vert s[k]\vert^{2}\Re\left\{ \nabla_{\bm{\eta}}^{\text{H}}(\alpha_0 e^{-j\frac{2\pi}{\lambda}\xi_{n}[k]})\nabla_{\bm{\eta}}(\alpha_0 e^{-j\frac{2\pi}{\lambda}\xi_{n}[k]})\right\} \nonumber 
\end{align}
in which the derivatives are given by 
\begin{equation}
\nabla_{\bm{\eta}}^{\text{T}}(\alpha_0 e^{-j\frac{2\pi}{\lambda}\xi_{n}[k]})=e^{-j\frac{2\pi}{\lambda}\xi_{n}[k]}\left[\begin{array}{c}
e^{j\psi}\\
j\alpha_0\\
-j\frac{2\pi}{\lambda}\alpha_0\frac{\partial\xi_{n}[k]}{\partial d}\\
-j\frac{2\pi}{\lambda}\alpha_0\frac{\partial\xi_{n}[k]}{\partial\theta}\\
-j\frac{2\pi}{\lambda}\alpha_0\frac{\partial\xi_{n}[k]}{\partial B}
\end{array}\right],\label{eq:derivatives}
\end{equation}
\Ra{where $\partial\xi_{n}[k]/\partial d=-kr_{f}$, $\partial\xi_{n}[k]/\partial\theta=n\Delta\sin\theta$,
$\partial\xi_{n}[k]/\partial B=kr_{f}$.} We find that, since $\Re\{j\}=0$,
$J_{1,i\neq1}=0$ (here $J_{i,i'}$ refers the entry in $\mathbf{J}(\bm{\eta})$
on row $i$, column $i'$). Hence, we can ignore $\rho$ when determining
the FIM of $\bm{\eta}=[\psi,d,\theta,B]^{\text{T}}$. We introduce
$\mathbf{e}_{i}$ as an all-zero vector with a 1 on index $i$, $\mathbf{b}=[1,0,-1]^{\text{T}}$,
and $\gamma=|\alpha_0|^{2}(2\pi/\lambda)^{2}/N_{0}$. Then 
\begin{align}
 & \mathbf{J}^{\text{S}}(\bm{\eta})=\gamma\mathbf{J}_{1}^{\text{S}}+\gamma\mathbf{J}_{2}^{\text{S}}+\gamma\mathbf{J}_{3}^{\text{S}},\label{eq:FIMstandardCase}
\end{align}
where
\begin{align}
\mathbf{J}_{1}^{\text{S}} & =\left(\frac{\lambda}{2\pi}\right)^{2}E_{K,0}E_{N,0}\mathbf{e}_{1}\mathbf{e}_{1}^{\text{T}}\\
\mathbf{J}_{2}^{\text{S}} & =E_{K,2}E_{N,0}r_{f}^{2}\left[\begin{array}{cc}
0 & \mathbf{b}^{\text{T}}\\
\mathbf{b} & \mathbf{0}_{3\times3}
\end{array}\right]\\
\mathbf{J}_{3}^{\text{S}} & =E_{K,0}E_{N,2}\Delta^{2}\sin^2\theta \,\mathbf{e}_{3}\mathbf{e}_{3}^{\text{T}},
\end{align}
in which $E_{K,i}$ $=$ $\sum_{k=-K/2}^{K/2}k^{i}|s[k]|^{2}$ and $E_{N,i}=\sum_{n=-N/2}^{N/2}n^{i}$.
The directions in which we obtain information are \emph{radially} (along
the line from the center of the BS array to the UE) and \emph{tangentially}
(orthogonal to the line between BS array center and UE). 
Transformation to the position domain is achieved as follows. With
$x=d\cos\theta$ and $y=d\sin\theta$, the FIM of $\mathbf{J}(\mathbf{x},B)$
is given by $\mathbf{J}(\psi,\mathbf{x},B)=\mathbf{T}^{\text{T}}\mathbf{J}(\bm{\eta})\mathbf{T}$
with Jacobian $\mathbf{T}$.\footnote{The Jacobian is given by $\mathbf{T}=$$[\mathbf{e}_{1}^{\text{T}};0,\,x/d,\,y/d,\,0;0,\,-y/d^{2},\,x/d^{2},\,0;\mathbf{e}_{4}^{\text{T}}]$
in which the ``$;$'' operator separates rows in a matrix. } Since $\psi$ does not depend on the other parameters, 
\begin{align}
 & \mathbf{J}(\mathbf{x},B)=\label{eq:FIMinPositionDomain}\\
 & \gamma E_{K,2}E_{N,0}r_{f}^{2}\mathbf{e}_{\mathbf{x}}\mathbf{e}_{\mathbf{x}}^{\text{T}}+\gamma E_{K,0}E_{N,2}\Delta^{2}y^{2}\frac{1}{\Vert\mathbf{x}\Vert^{4}}\mathbf{e}_{\mathbf{x},\perp}\mathbf{e}_{\mathbf{x},\perp}^{\text{T}}.\nonumber 
\end{align}
where $\mathbf{e}_{\mathbf{x}}=\left[x/d,\,y/d,\,1\right]^{\text{T}}$
and $\mathbf{e}_{\mathbf{x},\perp}=\left[-y/d,\,x/d,\,0\right]^{\text{T}}.$
Since $\mathbf{e}_{\mathbf{x}}$ is orthogonal to $\mathbf{e}_{\mathbf{x},\perp}$,
this decomposition shows that delay estimation provides radial information
with intensity $\gamma E_{K,2}E_{N,0}r_{f}^{2}$ and AOA estimation
provides tangential information with location-dependent intensity
$\gamma E_{K,0}E_{N,2}\Delta^{2}y^{2}/\Vert\mathbf{x}\Vert^{4}$.
Hence, AOA information is only useful for short distances. Moreover, the matrix $\mathbf{J}(\mathbf{x},B)$ is rank 2, so $B$
 and $\mathbf{x}$ are not identifiable.

\subsection{Localization Algorithm\label{subsec:Algorithm}}

We organize the observations $y_n[k]$ in a $(N+1)\times(K+1)$ matrix:
\begin{equation}
\mathbf{Y}=\alpha_0\,\mathbf{a}_{N+1}(\cos\theta)\mathbf{a}_{K+1}^{\text{H}}(\delta_{0}r_{f})\mathbf{S}+\mathbf{W}\label{eq:standardCaseVector}
\end{equation}
where $\mathbf{S}$ is a diagonal matrix containing the $K+1$ pilot symbols
and $\mathbf{a}_{M+1}(\cdot)$ is a vector of length $M+1$ with entries
$[\mathbf{a}_{M+1}(\beta)]_{m}=\exp\left(j2\pi\beta m/(M+1)\right),$
for $m=-M/2,\ldots,M/2.$ Similar to \cite{Shahmansoori2018}, we exploit the sparse nature of
$\mathbf{Y}$ by applying a 2D-FFT 
\begin{align}
\mathbf{Z}=\mathbf{F}_{N+1}\mathbf{Y}\mathbf{S}^{\text{H}}(\mathbf{S}\mathbf{S}^{\text{H}})^{-1}\mathbf{F}_{K+1}
\end{align}
to the observation $\mathbf{Y}\mathbf{S}^{\text{H}}(\mathbf{S}\mathbf{S}^{\text{H}})^{-1}$,
where the impact of the pilot symbols has been removed. \Rb{This also allows multiple users to be treated independently if there is no pilot contamination \cite{Rusek_MSP2013}. } Here, 
$\mathbf{F}_{M}$ denotes the $M\times M$ discrete Fourier transform
matrix. Higher accuracy can be achieved by zero-padding $\mathbf{Y}$
and applying larger FFT matrices. The peak of $|\mathbf{Z}|$ directly
provides us an estimate of $\cos\theta$ and $\delta_{0}r_{f}$. As
indicated by the FIM, the parameters are not identifiable, so we can
only localize the user when the bias $B$ is known. The complexity
of this method is of order $\mathcal{O}(NK\log KN)$. 

\section{Near-Field and Spatial Wideband Conditions}

\subsection{Narrowband Near-field Model and FIM}
\Rb{When signals arriving at different antennas are not resolvable in delay, the model \eqref{eq:basicSignalModelOmega} applies
with 
 (\ref{eq:signalPhaseGeneral}) specialized to}
\begin{align}
\xi_{n}[k] & =d_{n}+(kr_{f}-1)d-kr_{f}B.\label{eq:NarrowbandNearField}
\end{align}
We can now state the following result. 
\begin{thm}
In the case of narrowband near-field operation, the FIM of the parameter
$\bm{\eta}=[\psi,d,\theta,B]^{\text{T}}$ is
\[
\mathbf{J}^{\text{N}}(\bm{\eta})=\gamma\frac{A_{0}^{(0)}}{E_{N,0}}\mathbf{J}_{1}^{\text{S}}+\gamma\frac{A_{0}^{(0)}}{E_{N,0}}\mathbf{J}_{2}^{\text{S}}+\gamma\frac{A_{2}^{(2)}}{E_{N,2}}\mathbf{J}_{3}^{\text{S}}+\gamma\mathbf{J}_{4}^{\text{N}}+\gamma\mathbf{J}_{5}^{\text{N}},
\]
where $A_{i}^{(j)}=\sum_{n}\left({d}/{d_{n}}\right)^{j+2}n^{-i}$ and 
\[
\mathbf{J}_{4}^{\text{N}}=\frac{\lambda}{2\pi}E_{K,0}\left[\begin{array}{cc}
0 & \mathbf{j}^{\text{T}}\\
\mathbf{j} & \mathbf{0}_{3\times3}
\end{array}\right]
\]
\[
\mathbf{J}_{5}^{\text{N}}=E_{K,0}\left[\begin{array}{ccc}
0 & \mathbf{0}^{\text{T}} & 0\\
\mathbf{0} & \mathbf{C} & \mathbf{0}\\
0 & \mathbf{0}^{\text{T}} & 0
\end{array}\right]
\]
with $\mathbf{j}=[-\frac{\Delta}{d}\cos\theta A_{1}^{(1)}+A_{0}^{(1)}-A_{0}^{(0)},\,A_{1}^{(1)}\Delta\sin\theta,\,0]^{\text{T}}$,
$C_{1,1}=A_{0}^{(0)}+A_{0}^{(2)}-2(\frac{\Delta}{d}\cos\theta A_{1}^{(2)}+A_{0}^{(1)}-\frac{\Delta}{d}\cos\theta A_{1}^{(1)})+\frac{\Delta^{2}}{d^{2}}A_{2}^{(2)}\cos^{2}\theta$,
$C_{1,2}=C_{2,1}=\Delta\sin\theta A_{1}^{(2)}-\Delta\sin\theta A_{1}^{(1)}$,
$C_{2,2}=0$. 
\end{thm}
\begin{proof}
We readily find that $\partial\xi_{n}[k]/\partial d=(d-n\Delta\cos\theta)/d_{n}-1+kr_{f}$,
$\partial\xi_{n}[k]/\partial\theta=dn\Delta\sin\theta/d_{n}$,
$\partial\xi_{n}[k]/\partial B=-kr_{f}$. Substituting these derivatives in \eqref{eq:derivatives} and then in \eqref{eq:FIMexpression}, we obtain the desired result. 
\end{proof}

We observe that the first 3 components are similar to those in the
standard case (\ref{eq:FIMstandardCase}), up to a scaling. On the
other hand, $\mathbf{J}_{4}^{\text{N}}$ and $\mathbf{J}_{5}^{\text{N}}$
are due to the near-field propagation. In particular, $\mathbf{J}_{4}^{\text{N}}$
couples the channel phase $\psi$ with the UE distance $d$ and the AOA $\theta$. The diagonal
element $C_{1,1}$ in $\mathbf{J}_{5}^{\text{N}}$ provides additional
information on the distance, which allows $\mathbf{J}^{\text{N}}(\bm{\eta})$
to become full rank. This information is due to the dependence of
the curvature on the UE location, but not on the bias.  

\subsection{Spatial Wideband Far-field Model and FIM}

\Rb{Under spatial wideband far-field communication, the model \eqref{eq:basicSignalModelOmega} applies
with 
 (\ref{eq:signalPhaseGeneral}) specialized to}
\begin{align}
\xi_{n}[k] & =-n\Delta\cos\theta+k(d_{n}-B)r_{f}.\label{eq:WidebandFarField}
\end{align}
We can then state the following result. 
\begin{thm}
In the case of spatial wideband far-field operation, the FIM of the parameter
$\bm{\eta}=[\psi,d,\theta,B]^{\text{T}}$ is
\[
\mathbf{J}^{\text{W}}(\bm{\eta})=\gamma\frac{A_{0}^{(0)}}{E_{N,0}}\mathbf{J}_{1}^{\text{S}}+\gamma\mathbf{J}_{2}^{\text{W}}+\gamma\frac{A_{2}^{(2)}}{E_{N,2}}\mathbf{J}_{3}^{\text{S}}+\gamma\mathbf{J}_{4}^{\text{W}},
\]
where 
\begin{align*}
\mathbf{J}_{2}^{\text{W}} & =E_{K,2}r_{f}^{2}\left[\begin{array}{cccc}
0 & 0 & 0 & 0\\
0 & A_{0}^{(2)} & 0 & -A_{0}^{(1)}\\
0 & 0 & 0 & 0\\
0 & -A_{0}^{(1)} & 0 & A_{0}^{(0)}
\end{array}\right]
\end{align*}
\[
\mathbf{J}_{4}^{\text{W}}=E_{K,2}r_{f}^{2}\frac{\Delta\cos\theta}{d}\left[\begin{array}{cccc}
0 & 0 & 0 & 0\\
0 & A_{2}^{(2)}\frac{\Delta}{d}\cos\theta-2A_{1}^{(2)} & 0 & A_{1}^{(1)}\\
0 & 0 & 0 & 0\\
0 & A_{1}^{(1)} & 0 & 0
\end{array}\right]
\]
\end{thm}

\begin{proof}
We readily find that  $\partial\xi_{n}[k]/\partial d=kr_{f}(d-n\Delta\cos\theta)/d_{n}$,
$\partial\xi_{n}[k]/\partial\theta=n\Delta\sin\theta$, $\partial\xi_{n}[k]/\partial B=-kr_{f}$. Substituting these derivatives in \eqref{eq:derivatives} and then in \eqref{eq:FIMexpression}, we obtain the desired result. 
\end{proof}
We observe that the radial information in $\mathbf{J}_{2}^{\text{W}}$
is now scaled and that there is an additional term $\mathbf{J}_{4}^{\text{W}}$
that provides distance information with positive information $E_{K,2}r_{f}^{2}A_{2}^{(2)}\Delta^{2}\cos^{2}\theta/d^{2}$,
which is important for large $\Delta/d$. The information is larger
near the end-fire ($\theta\approx0$), as this is where the delay
spread is maximized. Note that the amount of information due to large
bandwidth is generally less than the amount of information due to
near-field. 

\subsection{Localization and Synchronization Algorithm\label{subsec:Algorithm-1}}
From the Fisher information analysis, we find that for small $\Delta$ (e.g., $\Delta = \lambda/2$) the amount of information increase due to near-field is more pronounced than due to spatial wideband. This is also confirmed by Fig.~\ref{fig:phasePlot}, where the nonlinear curvature of the phase across antennas (left Figure) is more significant than the difference in slope for different antennas (right Figure).  Based on this finding, we focus on the near-field case. 

The observation model in near-field is no longer of the form (\ref{eq:standardCaseVector})
so that a 2D-FFT will lead to multiple peaks. \Rc{A pure maximum likelihood
approach can be formulated, but leads to many local optima. Instead,
we propose to extend the method from Section \ref{subsec:Algorithm} with a simple  sub-arrays approach as in \cite{zhi2007near,myers_message_2019}, without aiming for an optimal solution.}
We divide the array (equivalently, the rows of $\mathbf{Y}\mathbf{S}^{\text{H}}(\mathbf{S}\mathbf{S}^{\text{H}})^{-1}$)
into non-overlapping  sub-arrays with $\tilde{N}$ elements.\footnote{Sub-array $\tilde{n}$ corresponding
to the observations at antenna $(\tilde{n}-1)\tilde{N}+1$ through
$\tilde{n}\tilde{N}$, with array center $\tilde{\mathbf{x}}_{\tilde{n}}=\mathbf{x}_{-N/2}+[\Delta((\tilde{n}-1)\tilde{N}+1+\tilde{N}/2,\,0]^{\text{T}}$. Here indexing $\tilde{n}$ starts at $1$. } The value of $\tilde{N}$
should be chosen to satisfy the following conditions: 
(i) \emph{far-field condition: }$\tilde{N}\le\sqrt{\bar{d}\lambda}/(2\Delta^{2})$,
so that the far-field assumption is valid per sub-array (here, $\bar{d}$
is an expected distance to the UE);
(ii) \emph{narrowband condition: }$\tilde{N}\ll c/(W\Delta)$, so that
paths are unresolved in the delay domain per sub-array. 
With these conditions (and ensuring that $\tilde{N}\ge1$), the method
from Section \ref{subsec:Algorithm} can be applied to each sub-array,
providing $\lfloor(N+1)/\tilde{N}\rfloor$ estimates 
\begin{align}
\hat{\theta}_{\tilde{n}} & =\arccos\left(\frac{x-\tilde{x}_{\tilde{n}}}{\Vert\mathbf{x}-\tilde{\mathbf{x}}_{\tilde{n}}\Vert}\right)+w_{\theta,\tilde{n}}, \label{eq:SUtheta}\\
\hat{\delta}_{\tilde{n}} & =\Vert\mathbf{x}-\tilde{\mathbf{x}}_{\tilde{n}}\Vert-B+w_{\delta,\tilde{n}} \label{eq:SUdelta},
\end{align}
where $w_{\theta,\tilde{n}}$ and $w_{\delta,\tilde{n}}$ are measurement
errors with variances $\sigma^2_{\theta,\tilde{n}}$ and   $\sigma^2_{\delta,\tilde{n}}$, due to the background noise and the finite resolution of the
FFTs. From these sub-array estimates, we can recover the UE position by intersecting the bearing lines and then the clock bias from the delay estimates. The complete procedure can be found in Algorithm \ref{alg:Alg1}. 
The complexity of the method is of order $\mathcal{O}(NK\log\tilde{N}K)$.
\begin{algorithm} 
\caption{Sub-array Localization and Synchronization}\label{alg:SLOS} \begin{algorithmic}[1] \Procedure{Localize-Near-Field}{$\mathbf{Y}$}
\label{alg:Alg1}
\State Determine $\tilde{N} \in \mathbb{N}_{\ge1}$: 
\begin{align*}
\tilde{N}\le\sqrt{\bar{d}\lambda}/(2\Delta^{2})~~\text{and}~~\tilde{N}\ll c/(W\Delta)
\end{align*}
\State Partition rows of $\mathbf{Y}$ into blocks of size $\tilde{N}$
\For{$\tilde{n}=1:\lfloor N+1/\tilde{N} \rfloor$}
\State Denote block $\tilde{n}$ by $\mathbf{Y}_{\tilde{n}}$
\State Estimate $\theta_{\tilde{n}}$ and $\delta_{\tilde{n}}$ from $\mathbf{Y}_{\tilde{n}}$ as in Section \ref{subsec:Algorithm}
\EndFor
\State Solve for $\mathbf{x}$:
\begin{align*}
\hat{\mathbf{x}}=\arg \min_{\mathbf{x}} \sum_{\tilde{n}=1}^{\lfloor N+1/\tilde{N} \rfloor} \frac{\left( \hat{\theta}_{\tilde{n}}-\arccos\left(\frac{x-\tilde{x}_{\tilde{n}}}{\Vert\mathbf{x}-\tilde{\mathbf{x}}_{\tilde{n}}\Vert}\right)\right)^2}{2 \sigma_{\theta,\tilde{n}}^2} 
\end{align*}
\State Solve for $B$:
\begin{align}
\hat{B} = \arg \min_{B} \sum_{\tilde{n}=1}^{\lfloor N+1/\tilde{N} \rfloor} \frac{\left( \hat{\delta}_{\tilde{n}}-\Vert\hat{\mathbf{x}}-\tilde{\mathbf{x}}_{\tilde{n}}\Vert-B)\right)^2}{2 \sigma_{\delta,\tilde{n}}^2} 
\end{align}
\State \textbf{return} $\hat{\mathbf{x}},\hat{B}$
\EndProcedure
\end{algorithmic}
\end{algorithm}

Note that in the narrowband far-field regime, $\tilde{N}=N+1$, so
that the method reverts to the one from Section \ref{subsec:Algorithm}.

\section{Numerical Results}\label{sec:NumericalResults}


We consider a nominal scenario at a carrier $f_{c}$ of 28 GHz ($\lambda\approx1.07$
cm), a bandwidth $W$ of 100 MHz, $c=0.3$ m/ns, $N_{0}=4.0049\times10^{-9}$
mW/GHz, a transmit power $P_{t}$ of 1 mW (with $\mathbb{E}\{|s[k]|^{2}\}=P_{t}/W$)
and $K+1=257$ subcarriers with QPSK pilots. The UE has bias $B=20$ m. The array has $N+1=129$ elements spaced at $\lambda/2$,
corresponding to a total size of 69.11 cm and a far-field distance
of 89 m.\footnote{Source code is available at \url{https://tinyurl.com/y3jybhdp}.}

\subsection{Fisher Information }

We will evaluate the position error bound (PEB\footnote{The PEB is defined from the $4 \times 4$ FIM $\mathbf{J}(\psi,\mathbf{x},B)$ as
$\sqrt{\text{trace}[\mathbf{J}^{-1}(\psi,\mathbf{x},B)]_{2:3,2:3}}$
and is expressed in m. }) for several models: the general (correct) model (\ref{eq:basicSignalModelOmega}), and three approximate models: the standard model (\ref{eq:SimplifiedModel}), the narrowband
near-field model (\ref{eq:NarrowbandNearField}), and the spatial wideband
far-field model (\ref{eq:WidebandFarField}), for known and unknown clock bias $B$. In Fig.~\ref{fig:PEBdistance}
we show the PEB as as a function of the distance $d$.
As expected, the far-field model is correct for distances
larger than about 8 m, while for shorter distances the near-field models
provide lower PEB. Moreover, at short distances, the PEB of the general model does not
depend on whether we know $B$, while for large distances, the PEB
quickly increases when $B$ is unknown. 
The results clearly show that joint synchronization and positioning in near-field can give good performance. 

In Fig.~\ref{fig:PEBDelta}, we show the PEB as a function of the
inter-antenna spacing. 
In this case, the PEB under the standard model does not depend on
$\Delta$, as it is mainly limited by the estimation of the distance.
The general model leads to lower PEB for large antenna spacing, and
larger PEB for small antenna spacing (for the case of unknown $B$).
Note that for very large $\Delta$, the PEB of the general model increases
due to the path loss. For both figures, we see that the
main benefit for small $d$ or large $\Delta$ comes from the near-field
information, not the spatial wideband information.

\begin{figure}
\begin{centering}
\includegraphics[width=1\columnwidth]{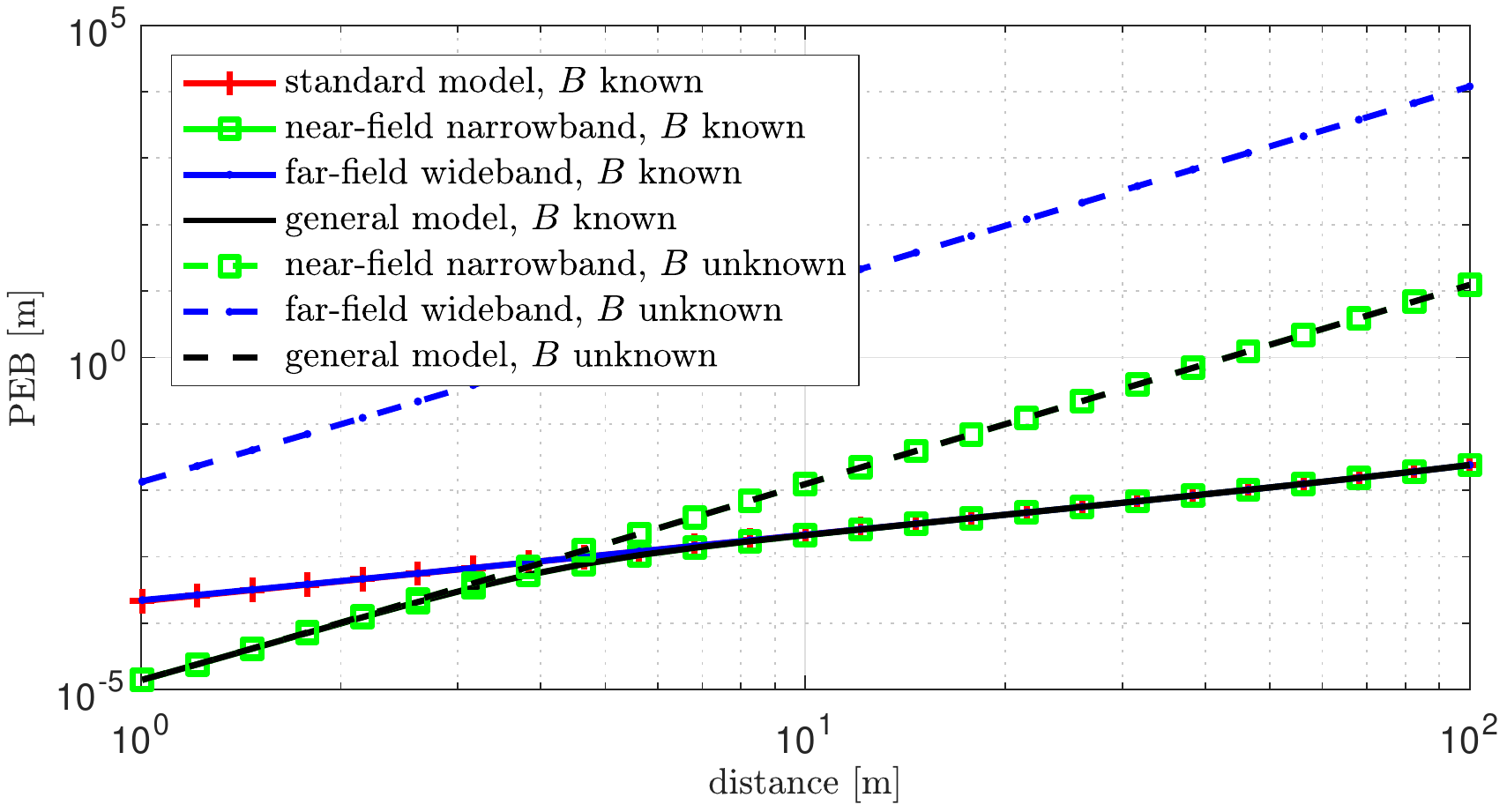}
\par\end{centering}
\caption{\label{fig:PEBdistance}PEB as a function of UE distance for known
and unknown clock bias $B$, with $\Delta=\lambda/2$. }
\end{figure}
\begin{figure}
\begin{centering}
\includegraphics[width=1\columnwidth]{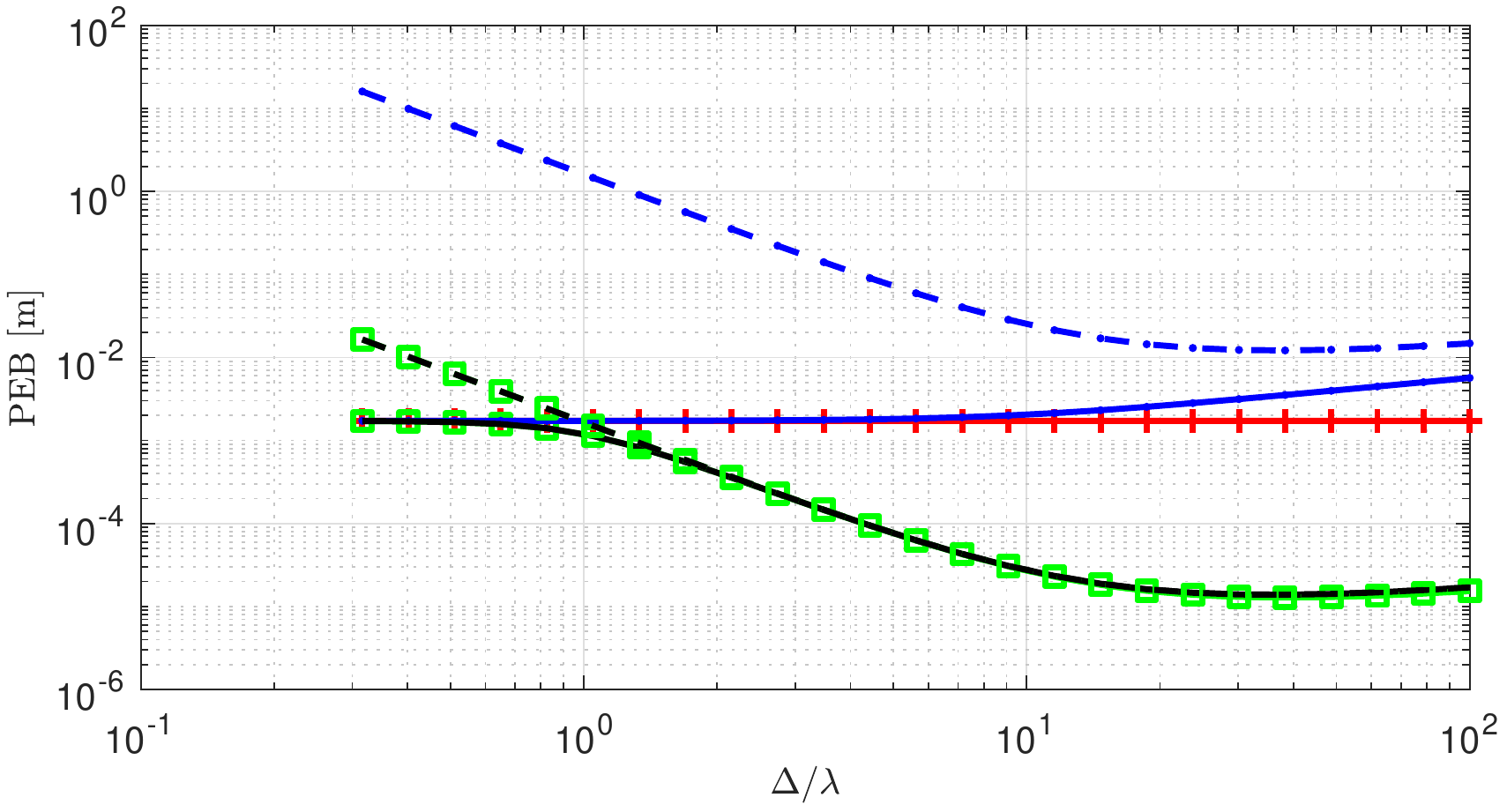}
\par\end{centering}
\caption{\label{fig:PEBDelta}PEB as a function of inter-antenna spacing for
known and unknown clock bias $B$, with $\mathbf{x}=[1\,\text{m},\,8\,\text{m}]^{\text{T}}$.
The legend is the same as in Fig.~\ref{fig:PEBdistance}.}
\end{figure}

\subsection{Algorithm}\label{sec:simsAlg}

We now evaluate the performance of the method described in Section
\ref{subsec:Algorithm-1}, whereby the AOAs are computed using a 2DFFT
with $10N$ points in the spatial domain and $K+1$ points in the
frequency domain. \EC{We vary $d$ with random $\theta \sim \mathcal{U}(\pi/4,3\pi/4)$ and place a scatterer with radar cross section of $10~\text{m}^2$ uniformly in the plane (this corresponds to a large scattering object). This enables us to evaluate the robustness to multipath. } For comparison purposes, the method from Section \ref{subsec:Algorithm} is also evaluated, assuming known bias. 
From Fig.~\ref{fig:Joint-localization-and}, we observe low position RMSE for distances
below 3 m. \Rb{The non-LOS (NLOS) path increases the RMSE compared to the LOS-only case, as it causes large outliers. Note that multipath appears as a second peak in the 2D-FFT and can thus be recovered and separated from the LOS path. However, which path corresponds to LOS is often harder to determine due to the poor delay resolution  (the resolution at $W=100$ MHz is only 3 m). }  Beyond 15 m distance, $\tilde{N} \to 1$,
so that the problem is no longer identifiable. The localization performance
is worse than the PEB from Fig.~\ref{fig:PEBdistance}, as the method
has not been optimized for accuracy. Moreover, the bias
estimate has orders of magnitude larger error, as it is based on low-quality
range estimates. This error can be further reduced by using larger FFTs along the frequency
dimension. The far-field method from Section \ref{subsec:Algorithm} with known bias is limited by the bandwidth and thus leads to worse performance for all $d$.
\begin{figure}
\includegraphics[width=1\columnwidth]{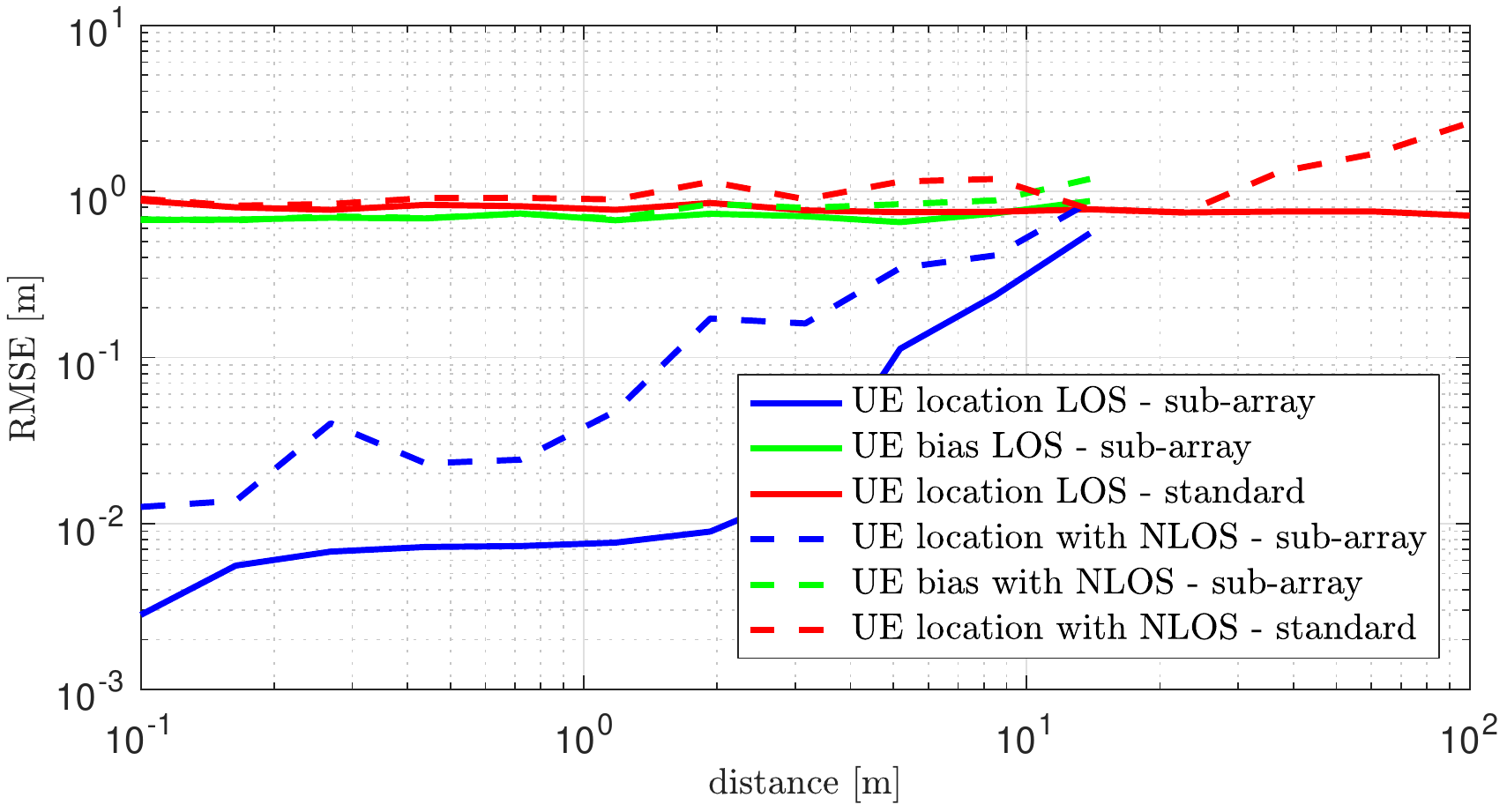}
\caption{\label{fig:Joint-localization-and}\Rb{Joint localization and synchronization
performance using sub-array processing vs standard processing with LOS only and with multipath (LOS + NLOS).} Bounds are not included, as they are loose since the algorithm is developed for proof-of-concept, not for ultimate accuracy. }
\end{figure}

\section{Conclusions}

When large arrays are used for positioning, near-field 
propagation must be taken into account. This presents challenges and
opportunities for the development of localization systems beyond 5G.
We have performed a Fisher information analysis and proposed a simple
joint localization and synchronization method for this regime. Our
results show that near-field propagation can be exploited in uplink and that the Fisher information provided from wavefront curvature is more pronounced than from spatial wideband. Immediate suggestions for future research are the inclusion of hybrid
combining at the BS, as in \cite{myers_message_2019}, the study of
downlink localization with a single receive antenna \cite{Fascista2019},
as well as the inclusion of \Rb{a more realistic propagation model \cite{Friedlander2019},
accounting for  coupling \cite{svantesson2001antennas}} and electromagnetic theory, as well as \Rd{jointly localizing and synchronizing multiple mobile users \cite{win2018theoretical}. }

\section*{Acknowledgment}
This research was supported, in part, by the Swedish Research Council under grant No.~2018-03701. 

\bibliographystyle{ieeetr}
\bibliography{ILSBibliography,references}

\end{document}